\documentclass[a4paper]{eptcs}
\usepackage{array,color}
\usepackage{stmaryrd}
\usepackage{amsmath,amssymb,amsthm}
\usepackage{enumitem}
\usepackage[draft]{fixme}
\usepackage{url}
\usepackage{xspace}
\usepackage{graphicx}
\usepackage{tabularx}

\newcommand{\first}{\textit{first}}
\newcommand{\length}{\textit{length}}

\newcommand{\Tree}{\textit{Tree}}

\newcommand{\K}{\textsf{K}}
\newcommand{\liff}{\leftrightarrow}
\newcommand{\A}{\mathsf{A}}

\newlength{\ksize}

\newlength{\rsize}

\newlength{\esize}

\newlength{\fsize}

\newcommand{\csetsc}[2]{\{#1 \mid #2\}}

\newcommand{\Mods}{\mathbf{Mods}}
\newcommand{\Mod}{\textsf{Mod}\xspace}
\newcommand{\Rel}{\textsf{Rel}\xspace}
\newcommand{\Atom}{\textsf{Atom}\xspace}
\newcommand{\Nom}{\textsf{Nom}\xspace}
\newcommand{\Prop}{\textsf{Prop}\xspace}
\newcommand{\ALit}{\textsf{ALit}\xspace}
\newcommand{\Ext}{\mathrm{Ext}}
\newcommand{\M}{\mathcal{M}}
\newcommand{\N}{\mathcal{N}}
\newcommand{\T}{\mathcal{T}}

\newcommand{\ssim}{\,\underline{\shortrightarrow}_\sigma\,}
\newcommand{\ssiminv}{\,\underline{\shortrightarrow}_{\sigma^{-1}}\,}
\newcommand{\sssim}{\,\underline{\shortrightarrow}_{\bar{\sigma}}\,}
\newcommand{\bisim}{\mathop{\underline{\leftrightarrow}}}

\newcommand{\boxr}[1]{[#1]}
\newcommand{\diamr}[1]{\tup{#1}}

\newcommand{\last}{\textit{last}}

 \newcommand{\cC}{\mathcal{C}\xspace}
 \newcommand{\cS}{\mathcal{S}\xspace}
 
 \newcommand{\cM}{\mathcal{M}\xspace}
 \newcommand{\cN}{\mathcal{N}}
 
 \newcommand{\cP}{\mathcal{P}\xspace}
 \newcommand{\Hl}{\mathcal{H}}

\newcommand{\cset}[1]{\{ #1 \}}
\newcommand{\tup}[1]{\langle #1 \rangle}

\newcounter{todoCount}
\stepcounter{todoCount}

\newtheorem{theorem}{Theorem}
\newtheorem{definition}{Definition}
\newtheorem{proposition}{Proposition}
\newtheorem{corollary}{Corollary}
\newtheorem{example}{Example}

\title{Symmetries in Modal Logics}

\author{Carlos Areces \institute{FaMAF\\ Universidad Nacional de
    C\'ordoba\\C\'ordoba, Argentina} \institute{CONICET, Argentina}
\and 
Guillaume Hoffmann\institute{FaMAF\\ Universidad Nacional de
    C\'ordoba\\C\'ordoba, Argentina} 
\and Ezequiel Orbe\institute{FaMAF\\ Universidad Nacional de
    C\'ordoba\\C\'ordoba, Argentina}\institute{CONICET, Argentina}}
\def\titlerunning{Symmetries in Modal Logics} 
\def\authorrunning{Areces, Hoffmann and Orbe}

\begin{document}
\maketitle

\begin{abstract}
  We generalize the notion of symmetries of propositional formulas in
  conjunctive normal form to modal formulas.  Our framework uses the
  coinductive models introduced in~\cite{arec:coin10} and, hence, the
  results apply to a wide class of modal logics including, for
  example, hybrid logics.  Our main result shows that the symmetries
  of a modal formula preserve entailment: if $\sigma$ is a symmetry of
  $\varphi$ then $\varphi \models \psi$ if and only if $\varphi
  \models \sigma(\psi)$.
\end{abstract}

\section{Symmetries in Automated Theorem Proving}\label{sec:intro}

Many concrete, real life problems present symmetries.  For instance,
if we want to know whether trying to place three pigeons in two
pigeonholes results in two occupying the same nest, it does not really
matter which of all pigeons gets in each pigeonhole.  Starting by
putting the first pigeon to the first pigeonhole is the same as if we
put the second one in it.  In mathematical and common-sense reasoning
these kinds of symmetries are often used to reduce the difficulty of
reasoning --- one can analyze in detail only one of the symmetric
cases and generalize the result to the others.  The exact same is done
in propositional theorem proving. Many problem classes and, in
particular, those arising from real world applications, display a
large number of symmetries; and current SAT solvers take into account
these symmetries to avoid exploring duplicate branches of the search
space. In the last years there has been extensive research in this
area, focusing on how to define symmetries, how to detect them
efficiently, and how SAT solvers can better profit from
them~\cite{Prasad05asurvey}.

Informally, we can define a symmetry of a discrete object as a
permutation of its components that leaves the object, or some aspect
of it, intact (think of the rotations of a spatial solid).  In the context
of SAT solving we can formally define a symmetry as a permutation of
the variables (or literals) of a problem that preserves its
structure and, in particular, its set of solutions.  Depending on
which aspect of the problem is kept invariant, symmetries are
classified in the literature into semantic or
syntactic~\cite{Benhamou:1992vx}.  Semantic symmetries are intrinsic
properties of a Boolean function that are independent of any
particular representation, i.e., a permutation of variables that does
not change the value of the function under any variable
assignment. Syntactic symmetries, on the other hand, correspond to the
specific algebraic representation of the function, i.e., a permutation
of variables or literals that does not change the representation. A
syntactic symmetry is also a semantic symmetry, but the converse does
not always hold. 

In~\cite{Krishnamurthy:1985ug}, Krishnamurthy used symmetries in the
context of SAT solving. In this article, the notions of \emph{global}
and \emph{local} symmetries as inference rules are used to strengthen
resolution-based proof systems for propositional logic, showing that
they can shorten the proofs of certain difficult propositional
problems like the pigeonhole principle. Since then, many articles discuss how to detect and exploit symmetries. Most of them can be
grouped into two different approaches: static symmetry breaking and
dynamic symmetry breaking.  In the first
approach~\cite{Crawford:1992wz,Crawford:1996wa,Aloul:2002ww},
symmetries are detected and eliminated from the problem statement
before the SAT solver is used. They work as a preprocessing step.  In
contrast, dynamic symmetry breaking~\cite{Brown:1989uw,Benhamou:1992vx,Benhamou:1994tr} detects
and breaks symmetries during the search space exploration.  The first
approach can be used with any theorem prover; the second is prover 
dependent but it can take
advantage of symmetries that emerge during search.
Despite their differences they share the same goal: to identify
symmetric branches of the search space and guide the SAT solver away
from symmetric branches already explored.  A third alternative was
introduced in~\cite{Benhamou:uv}, which combines symmetry reasoning
with clause learning~\cite{ryan} in Conflict-Driven Clause Learning
SAT solvers~\cite{Een:2004uh}. The idea is to augment clause
learning by using the symmetries of the problem to learn the
symmetric equivalents of conflict-induced clauses. This approach is
particularly appealing as it does not imply major modifications
to the search procedure and the required modifications to the clause
learning process are minor.

Symmetries have been extensively investigated and successfully
exploited for propositional logic SAT and some results involve other
logics, see~\cite{Audemard:2002ti,Audemard:2004up,fontaine2011}.  To the best
of our knowledge, symmetries remain largely unexplored in automated
theorem proving for modal logics.  

In this paper, we generalize the notion of symmetries to modal formulas in
conjunctive normal form for different modal logics including the basic
modal language over different model classes (e.g., reflexive, linear
or transitive models), and logics with additional modal operators
(e.g., universal and hybrid operators).  The main result of the
article shows that symmetries of a modal formula preserve entailment:
if $\sigma$ is a symmetry of $\varphi$ then $\varphi \models \psi$ if
and only if $\varphi \models \sigma(\psi)$.  In cases where the modal
language has a tree model
property, we can actually use a more flexible notion of symmetry that
enables different permutations to be applied at each modal depth.  
We also present a method to detect the symmetries of modal
formulas in conjunctive normal form. This method reduces the symmetry
detection problem to the graph automorphism problem. A  general graph
construction algorithm, suitable for many modal logics, is presented.
In order to tackle a broad range of modal languages that may or may not enjoy
the tree model property, we use in our work  the semantics provided by
coinductive modal models~\cite{arec:coin10} instead of the more familiar Kripke relational
semantics. Coinductive modal models provide a homogeneous framework
to investigate different modal languages at a greater level of
abstraction. A consequence of this is that results obtained in the
coinductive framework can be easily extended to concrete modal
languages by just giving the appropriate definition of the model
classes and fixing some parameters.

In Section~\ref{sec:defs} we present modal logics and coinductive
modal models.  In Section~\ref{sec:lit-symm} we define modal
symmetries, together with the appropriate notion of simulation to show
that symmetries preserve modal entailment.  In
Section~\ref{sec:layering} we introduce layered permutations and show
that they can be used when the modal logic has the adequate notion of
the tree model property. In Section \ref{sec:detection} we present a
graph construction algorithm to detect symmetries in modal formulas
and prove its correctness. We draw our conclusions and discuss future
research in Section~\ref{sec:conclusion}.


\section{Modal Logics and Coinductive Models} \label{sec:defs}

In what follows, we will assume basic knowledge of classical modal
logics and refer the reader to~\cite{MLBOOK,BBW06} for technical
details.  The coinductive framework for modal logics was introduced
in~\cite{arec:coin10} to investigate normal forms for a wide number of
modal logics.  Its main characteristic is that it allows the
representation of different modal logics in a homogeneous form.
For a start, the set of formulas is as for the basic (multi) modal
logic.

\begin{definition}[Modal formula]\label{modalformula}
  A modal signature is a pair $\tup{\Atom, \Mod}$ where \Atom
  and \Mod are two countable, disjoint sets. We usually assume that
  $\Atom$ is infinite. The set of modal formulas over $\langle
  \Atom,$ $\Mod \rangle$ is defined as
$$
\varphi ::= a \mid \lnot \varphi \mid \varphi \lor \varphi \mid
\boxr{m}\varphi,
$$

\noindent
for $a \in \Atom$, $m \in \Mod$. $\top$ and $\bot$ stand for 
an arbitrary tautology and contradiction, respectively. Connectives such
as $\land, \to$ and $\diamr{m}$, are defined as usual.
\end{definition}

We will define a symmetry as a permutation of literals that
preserve the structure of formulas in conjunctive normal form (CNF).

\begin{definition}[Literals and modal CNF]
  A \emph{literal} $l$ is either an atom $a$ or its negation $\neg a$.  The
  set of literals over $\Atom$ is $\ALit = \Atom \cup \{\neg a \mid a
  \in \Atom\}$. 


  A modal formula is in \emph{modal conjunctive normal form (modal
    CNF)} if it is a conjunction of modal CNF clauses. A \emph{modal
    CNF clause} is a disjunction of atoms and modal literals. A
  \emph{modal literal} is a formula of the form $\boxr{m} C$ or $\neg
  \boxr{m} C$ where $C$ is a modal CNF clause.  Every modal formula
  can be transformed into an equisatisfiable formula in modal CNF in
  polynomial time (see~\cite{arec:coin10,sebastiani} for details).
\end{definition} 

A formula in modal CNF can be represented as a set of modal CNF clauses
(interpreted conjunctively), and each clause can be
represented as a set of atom and modal literals
(interpreted disjunctively).  With the set representation we can
disregard the order and multiplicity in which clauses and literals
appear.  This will be important when we define symmetries below. In the rest
of the paper we will assume that modal formulas are in modal CNF, and
we will refer to them as modal CNF formulas.
  
\begin{example}
The modal formula $\varphi = \diamr{m}(p \wedge q \wedge p) \wedge
\boxr{m} \neg r$ is equisatisfiable to the modal CNF formula $ \varphi' =\{ \{ \neg \boxr{m}
\{ \neg p, \neg q \} \} , \{ \boxr{m} \{\neg r \}\} \}$.  
\end{example}

Up to now, we have not departed from the standard presentation of
classical modal logic in important ways.  The main change introduced
by the coinductive approach is with the definition of model and
semantic conditions.

\begin{definition}[Models]\label{models}
  Let $\cS = \tup{\Atom, \Mod}$ be a modal signature and $W$ be a
  fixed, non-empty set. $\Mods_{W}$, the class of all models with
  domain $W$, for the signature $\cS$, is the class  of all tuples
  $\tup{w, W, V, R}$ such that $w \in W$, $V(v) \subseteq \Atom$ for
  all $v \in W$, and
$$
R(m,v) \subseteq \Mods_{W} \mbox{\ for $m \in \Mod$ and $v \in W$.}
$$

\noindent
Given a model $\cM = \tup{w, W, V, R}$ we will say that $w$ is the
point of evaluation and denote it as $w^\cM$, $W$ is the domain and
denote it as $|\cM|$, $V$ is the modal valuation and denote it as $V^\cM$,
and $R$ is the accessibility relation and denote it as $R^\cM$.
$\Mods$ denotes the class of all models over all domains, $\Mods =
\bigcup_{W} \Mods_W$.

Given $\cM \in \Mods_{W}$, let $\Ext(\cM)$, the \emph{extension of
  $\cM$}, be the smallest subset of $\Mods_{W}$ that contains $\cM$
and is such that if $\cN \in \Ext(\cM)$, then $R^{\cN}(m,v) \subseteq
\Ext(\cM)$ for all $m \in \Mod$, $v \in W$.
\end{definition}

The definition of a coinductive modal model is similar to the usual
definition of a Kripke pointed model. The difference lies in the way
the accessibility relation is defined. In particular, for each $m$ and
each state $w$, $R(m,w)$ is defined as the set of (potentially
different) models accessible from $w$ through the $m$ modality. Observe
that for each $W$, $\Mods_{W}$ is well-defined (coinductively), and so
does $\Mods$, the class of all models. Our results will apply not only
to $\Mods$ but to many of its subclasses.  We will be interested in
classes which are \emph{closed under accessibility relations}
(\emph{closed classes} for short): $\cM \in \cC$ implies $\Ext(\cM)
\subseteq \cC$.  In the rest of the paper we will only consider
classes of models closed under accessibility relations.

\begin{example}
  Consider the pointed Kripke model in
  Figure~\ref{fig:model-diff}a, and its equivalent coinductive modal model in
  Figure~\ref{fig:model-diff}b. The point of evaluation in each model is circled. 
The main difference is that the relation of a coinductive model leads to another
coinductive model, whereas in a Kripke model the relation leads
to another point of the same model.

\begin{figure}[ht]
  \begin{center}
    \begin{minipage}[ht]{.35\linewidth}
      \begin{center}
        \includegraphics[scale=0.6,keepaspectratio=true]{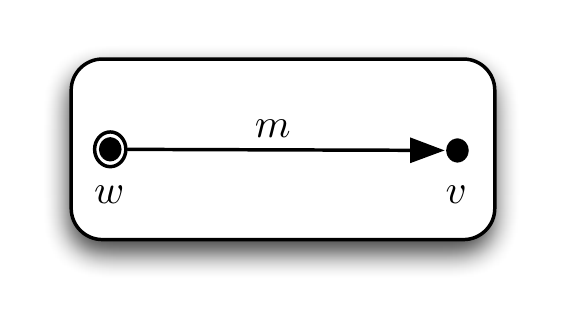}
          
        (a)
      \end{center}
    \end{minipage}
    \begin{minipage}[ht]{.55\linewidth}
      \begin{center}
        \includegraphics[scale=0.6,keepaspectratio=true]{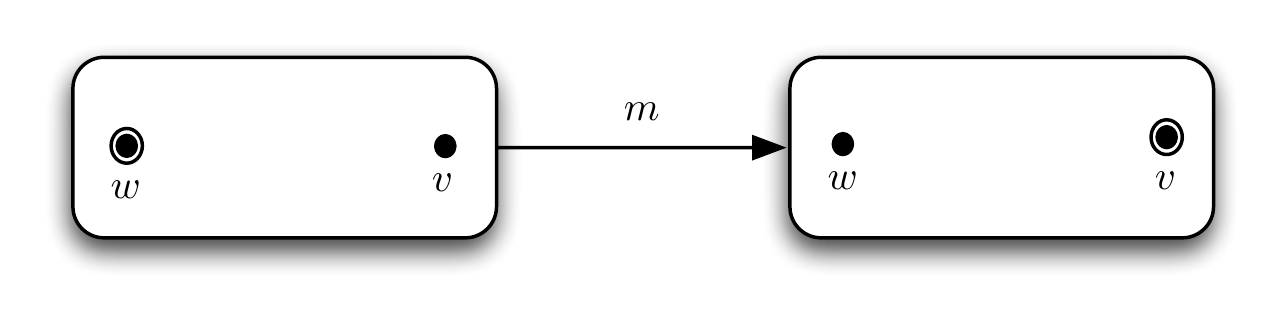}
          
        (b)
      \end{center}
    \end{minipage}
  \end{center}
  \caption{a) A Kripke model. b) Equivalent coinductive  model.  \label{fig:model-diff}}
\end{figure}






\end{example}

We are now ready to introduce the definition of the satisfiability
relation $\models$.

\begin{definition}[Semantics]\label{def:sem}
  Let $\varphi$ be a formula in modal CNF and $\cM = \tup{w,W,V,R}$ a model in
  $\Mods$. We define $\models$ for modal CNF formulas, clauses and literals
  as
  \begin{center}
    \begin{tabular}{lll}
      $\cM \models \varphi$ &  \mbox{ iff }
      & for all clauses $C \in \varphi$ we have $\cM \models C$\\  
      $\cM \models C$ &  \mbox{ iff }
      &  there is some literal $l \in C$ such that $\cM \models l$\\
      $\cM \models a$ & \mbox{ iff } & $a \in V(w)$ \mbox{ for } $a \in \Atom$\\
      $\cM \models \neg a$ & \mbox{ iff } & $a \not \in V(w)$ \mbox{ for } $a \in \Atom$\\
      $\cM \models \boxr{m} C$ &  \mbox{ iff } & $\cM' \models C$, for all $\cM' \in R(m,w)$\\  
      $\cM \models \neg \boxr{m} C$ &  \mbox{ iff } & $\cM \not \models \boxr{m}C$.
    \end{tabular}
  \end{center}

\noindent
For $\cC$ a class of models, we write $\cC \models \varphi$ whenever
$\cM \models \varphi$ for every $\cM$ in $\cC$, and we say that
$\Gamma_{\cC} = \cset{\varphi \mid \cC \models \varphi}$ is the
\emph{logic} defined by $\cC$.

The set of models in $\cC$ of a formula $\varphi$ is the set $\Mod_\cC(\varphi) 
=\csetsc{\cM}{\cM \in \cC \mbox{ and } \cM \models \varphi}$ (when
$\cC$ is clear from the context we will just write $\Mod(\varphi)$).  We
say that $\psi$ \emph{can be inferred from} $\varphi$ in $\cC$ and
write $\varphi \models_\cC \psi$ if $\Mod_\cC(\varphi) \subseteq
\Mod_\cC(\psi)$.
\end{definition}

As shown in~\cite{arec:coin10}, the logic $\Gamma_\Mods$ (generated by
the class of all possible models) coincides with the basic multi-modal
logic \textsf{K}.  By properly restricting the model class we can
capture different modal logics.  Let us call a predicate $P$ on models
a \emph{defining condition} for a class $\cC$ whenever $\cC$ is  such
that $\cM \in \cC$ if and only if $P(\cM)$
holds.  Consider the signature $\cS = \tup{\Atom, \Mod}$ where $\Atom
= \Prop \cup \Nom$, $\Mod = \Rel \cup \cset{\A} \cup \csetsc{@_i}{i
  \in \Nom}$; and $\Prop = \cset{p_1,p_2,\ldots}$, $\Nom =
\cset{n_1,n_2,\ldots}$ and $\Rel = \cset{r_1,r_2,\ldots}$ are mutually
disjoint, countable infinite sets. In what follows, we will usually be
interested in sub-languages of the language defined over $\cS$ by
Definition~\ref{modalformula}.

\begin{figure}[ht]
  \begin{center}
    \begin{tabular}{|c|l@{\ }c@{\ }l@{\ }c@{\ }l|} \hline

      Class & \multicolumn{5}{l|}{Defining condition} \\ \hline

      $\cC_m^\K$ & $\cP_{m}^{\K}(\cM)$ & $\Longleftrightarrow$ &
      $R^{\cM}(m,w)$ & $\subseteq$ & $\cset{\tup{v,|\cM|,V^{\cM},R^{\cM}} \mid v \in |\cM|}, m \in \Rel$  \\ \hline

      $\cC_{\A}$ & $\cP_\A(\cM)$ & $\Longleftrightarrow$ &
      $R^{\cM}(\A, w)$  &  = & $\cset{\tup{v, |\cM|,V^{\cM}, R^{\cM}} \mid v \in |\cM|}$\\ \hline

      $\cC_{@_i}$ & $\cP_{@_i}(\cM)$ & $\Longleftrightarrow$ &
      $R^{\cM}(@_i,w)$ & = & $\cset{\tup{v,|\cM|,V^{\cM},R^{\cM}} \mid i \in V(v)},  i \in \Nom$\\\hline
      $\cC_{\mbox{\tiny \Nom}}$ & $\cP_{\mbox{\tiny \Nom}}(\cM)$ & $\Longleftrightarrow$ &
      \multicolumn{3}{@{}l|}{$\{w \mid i \in V^{\cM}(w)$\} is a singleton, $\forall i \in \Nom$} \\ \hline 
    \end{tabular}
  \end{center}
  \vspace*{-.3cm}
  \caption{Defining conditions for different modal
    logics}\label{fig:classes}
\end{figure}

Figure~\ref{fig:classes} introduces a number of closed model classes
by means of their defining conditions. Observe that $\cP_{m}^{\K}$ is
true for a model $\cM$ if every successor of $w^{\cM}$ is identical to
$\cM$ except perhaps on its point of evaluation. We call $m$ a
\emph{relational modality} when it is interpreted in $\cC_{m}^{\K}$
because over this class they behave as classical relational
modalities~\cite{arec:coin10}.

We can capture different modal operators, like the ones from hybrid
logics~\cite{arec:hybr05b}, by choosing the proper class of models.
Predicates $\cP_\A$ and $\cP_{@_i}$, for instance, impose conditions
on the point of evaluation of the accessible models restricting the evaluation
to the class of models where the relation is, respectively, the total
relation ($\forall x y . R(x,y)$) and the `point to all $i$' relation
($\forall x y . R(x,y) \liff i(y)$). Observe that whenever the atom
$i$ is interpreted as a singleton set, the `point to all $i$' relation
becomes the usual `point to $i$' relation ($\forall x y . R(x,y) \liff
y = i$) of hybrid logics.  Finally, predicate $\cP_{\Nom}$ turns
elements of $\Nom$ into nominals, i.e., true at a unique element of
the domain of the model.  

An interesting feature of this setting is that we can express the
combination of modalities as the intersection of their respective
classes.  For example, $\cC_{\Hl(@)}$, the class of models for the
hybrid logic $\Hl(@)$, can be defined as follows:
$$
\begin{array}{c}
  \cC_{\Hl(@)}  =  \cC_{\Nom} \cap \cC_@  \cap \cC_{\Rel}\mbox{, where }\\
  \cC_{@}   =  \bigcap_{i \in \Nom}\cC_{@_i} \mbox{, and }
  \cC_{\Rel}  =  \bigcap_{m \in \Rel}\cC^{\K}_{m}.
\end{array}
$$

The crucial characteristic of the coinductive approach is that all these 
different modal operators are captured using the same semantic condition 
introduced in Definition~\ref{def:sem}. All the details defining each 
particular operator are now introduced as properties of the accessibility 
relation. As a result, a unique notion of bisimulation is sufficient to 
cover all of them.  

\begin{definition}[Bisimulations]\label{def:bisim}
  Given two models $\cM$ and $\cM'$ we say that \emph{$\cM$ and $\cM'$
    are bisimilar} (notation $\cM \bisim \cM'$) if $\cM \mathop{Z}
  \cM'$ for some relation $Z \subseteq \Ext(\cM) \times \Ext(\cM')$
  such that whenever $\tup{w,W,V,R} \mathop{Z}$ $\tup{w',W',V',R'}$ we
  have the following properties:
  \begin{itemize}
  \item \textbf{Harmony:} $a \in V(w)$ iff $a \in V'(w')$, for all $a
    \in \Atom$.
  \item \textbf{Zig:} $\cN \in R(m,w)$ implies $\cN \mathop{Z} \cN'$
    for some $\cN' \in R'(m,w')$.

  \item \textbf{Zag:}$\cN' \in R'(m,w')$ implies $\cN \mathop{Z} \cN'$
    for some $\cN \in R(m,w)$.
  \end{itemize}
  Such $Z$ is called a \emph{bisimulation between $\cM$ and $\cM'$}.
\end{definition}

The classic result of invariance of modal formulas under
bisimulation~\cite{MLBOOK} can easily be proved.

\begin{theorem}\label{invariance}
  If $\cM \bisim \cM'$, then $\cM \models \varphi$ iff $\cM' \models
  \varphi$, for all $\varphi$.
\end{theorem}

As stated in~\cite{arec:coin10},  this 
general notion of bisimulation works for every modal logic definable
as a closed subclass of $\Mods$.


 \section{Modal Symmetries}\label{sec:lit-symm}

We will show that consistent symmetries for modal
 formulas behave similarly as in the propositional case and, hence,
 could assist in modal theorem proving.  Let us start by introducing
 the basic notions.

 \begin{definition}[Complete, Consistent and Generated sets of
   literals]
     A set of literals $L$ is \emph{complete} if for each $a \in \Atom$
  either $a \in L$ or $\neg a \in L$. It is \emph{consistent} if for
  each $a \in \Atom$ either $a \not \in L$ or $\neg a \not \in L$. Any
  complete and consistent set of literals $L$ defines a unique
  valuation $v \subseteq \Atom$ as $a \in v$ if $a
  \in L$ and $a \not \in v$ if $\neg a \in L$.  For $S \subseteq
  \Atom$, the \emph{consistent and complete set of literals generated
    by} $S$ (notation $L_{S}$) is $S \cup \{\neg a \mid a \in \Atom
  \backslash S \}$.  
 \end{definition}

\begin{definition}[Permutation]
  A \emph{permutation} is a bijective function $\sigma : \ALit \mapsto
  \ALit$. For $L$ a set of literals, $\sigma(L) = \{\sigma(l) \mid l
  \in L\}$.
\end{definition}

In this work, we only deal with permutations defined over finite sets
of literals, namely, those occurring in the formula $\varphi$ under
consideration. This restricts us to finite groups of
symmetries~\cite{fraleigh2003first}. This kind of permutations can be
succinctly defined using cyclic notation, e.g., $\sigma=(p \ \neg
q)(\neg p \ q)$ is the permutation that makes $\sigma(p)=\neg q$,
$\sigma(\neg q)= p$, $\sigma(\neg p)= q$ and $\sigma(q)= \neg p$ and
leaves unchanged all other literals; $\sigma=(p\ q\ r)(\neg p\ \neg q\
\neg r)$ is the permutation $\sigma(p)=q$, $\sigma(q)=r$ and
$\sigma(r)=p$ and similarly for the negations. Finally, for $n \in
\mathbb{Z}_{\geq 1}$ and $\sigma$ a permutation, we denote the composition of
$\sigma$ with itself $n$ times by $\sigma^n$. $\sigma^0$ denote the
identity permutation, $\sigma^{-1}$ the inverse of $\sigma$, and
$\sigma^{-n}$, for $n \in \mathbb{Z}_{\leq 1}$ is the $n$-times composition of
$\sigma^{-1}$ with itself.


Because, in our language, atoms may occur in some modalities (like $@_i$) we should take
some care when we apply permutations to modal formulas.  We will say
that a modality is \emph{indexed by atoms} if its definition depends
on the value of an atom. If $m$ is indexed by an atom $a$ we will
sometimes write $m(a)$.

\begin{definition}\label{def:permutation}[Permutation of a formula]
  Let $\varphi$ be a modal CNF formula and $\sigma$ a permutation.
  We define $\sigma(\varphi)$ recursively:
$$
\begin{array}{rcll}
  \sigma(\varphi) & = & \csetsc{\sigma(C)}{C \in \varphi} & \mbox{ for
    $\varphi$ a modal CNF formula}\\ 
  \sigma(C) & = & \csetsc{\sigma(A)}{A \in C} & \mbox{ for $C$ a modal
    CNF clause}\\ 
  \sigma(\boxr{m}C) & = & \boxr{\sigma(m)}\sigma(C) &\\ 
  \sigma(\neg\boxr{m}C) & = & \neg\boxr{\sigma(m)}\sigma(C) &\\ 
\end{array}
$$
where $\sigma(m) =\sigma(m(a)) = m(\sigma(a))$ if $m$ is indexed by $a$, and
$\sigma(m) = m$ otherwise.
\end{definition}
  
\begin{definition}
  A permutation $\sigma$ is \emph{consistent} if for every
  literal $l$, $\sigma (\neg l) = \neg \sigma (l)$. A
  permutation $\sigma$ is a \emph{symmetry} for $\varphi$ if $\varphi
  = \sigma(\varphi)$, when conjunctions and disjunctions in $\varphi$
  are represented as sets.
\end{definition}

\begin{example} Trivially, the identity permutation $\sigma(l) = l$ is
  a consistent symmetry of any formula $\varphi$. More interestingly,
  consider $\varphi = \{\{\neg p, r\}$, $\{q, r\}$, $\{r,[m]\{\neg p,$
  $q\}\}\}$, then the permutation $\sigma = (p\ \neg q)(\neg p\ q)$
  is a consistent symmetry of $\varphi$.
\end{example}

Now, since a permutation over literals can be lifted to transform
some formula $\varphi$ into another formula $\sigma(\varphi)$,
we also want to consider permutations applied to models.
Indeed, if $\varphi$ is true in some $\M$, we intuitively want
$\sigma(\varphi)$ to be true in some model obtained from lifting
$\sigma$ to $\M$.
Thus the next step is
to define the notion of applying permutations to models.

\begin{definition}[Permutation of a model]\label{def:sigma-model2}
  Let $\sigma$ be a permutation and $\M = \langle w, W, V, R \rangle$
  a model. Then $\sigma(\M) = \langle w, W, V',R' \rangle$, where,
$$
\begin{array}{rll}
  V'(v) &=& \sigma (L_{V(v)}) \cap \Atom \quad \mbox{for all } v \in W
  \mbox{, and, }\\
  R'(m,v)&=&\{ \sigma(\N) \mid \N \in R(\sigma(m),v)\} \quad \mbox{for all
  } m \in \Mod \mbox{ and } v \in W.
\end{array}
$$
For $M$ a set of models, $\sigma(M) = \{\sigma(\M) \mid \M \in M\}$.
\end{definition}

The main ingredient to prove that symmetries preserve entailment is the
relation between models that we call $\sigma$-simulation.  
 
\begin{definition}[$\sigma$-simulation]\label{def:sigma-sim}
  Let $\sigma$ be a permutation. A \emph{$\sigma$-simulation} between
  models $\M = \langle w, W,V,R \rangle$ and $\M' = \langle
  w',W',V',R' \rangle$ is a non-empty relation $Z \subseteq \Ext(\M)
  \times \Ext(\M')$ that satisfies the following conditions:

  \begin{itemize}
  \item \textbf{Root:} $\M Z \M'$.

  \item \textbf{Harmony:} $l \in L_{V(w)}$ iff $\sigma (l) \in
    L_{V'(w')}$.
  \item \textbf{Zig:} If $\M Z \M'$ and $\N \in R(m,w)$ then $\N Z \N'$ for some $\N'
    \in R'(\sigma(m),w')$.

  \item \textbf{Zag:} If $\M Z \M'$ and $\N' \in R'(m,w')$ then $\N Z \N'$ for some
    $\N \in R(\sigma^{-1}(m),w)$.
  \end{itemize}

  We say that two models $\M$ and $\M'$ are $\sigma$-similar (notation
  $\M \ssim \M'$) if there is a $\sigma$-simulation $Z$ between them.
\end{definition}

Notice that while $\M \ssim \M'$ implies $\M' \ssiminv \M$, the relation
$\ssim$ is not symmetric (in particular $\sigma$ might differ from
$\sigma^{-1}$).
From the definition of $\sigma$-simulations it intuitively follows
that while they do not preserve validity of modal formulas (as is the
case with bisimulations) they do preserve validity of
\emph{permutations} of formulas.

\begin{proposition}\label{prop:lit2}
  Let $\sigma$ be a consistent permutation, $\varphi$ a modal CNF formula
  and $\M=\tup{w,W,V,R}$, $\M'=\tup{w',W',V',R'}$ models such that
  $\M\ssim \M'$.  Then $\M \models \varphi$ iff $\M' \models \sigma
  (\varphi)$.
\end{proposition}

\begin{proof}
  The proof is by induction on
  $\varphi$. 
  Base Case: Suppose $\varphi = a$ then, $\M \models a$ iff $a \in
  V(w)$ iff $a \in L_{V(w)}$ iff, by definition of
  $\sigma$-simulation, $\sigma(a) \in L_{V'(w')}$ iff $\M' \models
  \sigma(a)$.

  Suppose $\varphi = \neg a$ then, $\M \models \neg a$ iff $a \not\in
  V(w)$ iff $\neg a \in L_{V(w)}$ iff, by definition of
  $\sigma$-simulation, $\sigma(\neg a) = \neg \sigma(a) \in
  L_{V'(w')}$ iff $\sigma(a) \not\in V'(w')$ iff $\M' \models \neg
  \sigma(a)$.
  
  When $\varphi=C$, with $C$ a clause or a conjunction of
  clauses, the proof follows by induction directly.

  Inductive Step: Suppose $\varphi = \boxr{m} \psi$. Then $\M \models
  \boxr{m} \psi$ iff $\N \models \psi$ for all $\N \in R(m,w)$. Given
  that $\M\ssim \M'$, by Zig we know that for all $\N$ exist $\N'$
  such that $\N \ssim \N'$ and $\N' \in R'(\sigma(m),w')$. Then, by
  inductive hypothesis, $\N' \models \sigma(\psi)$ for all $\N' \in
  R'(\sigma(m),w')$ iff $\M' \models \boxr{\sigma(m)}\sigma(\psi)$.
  Then, by Definition \ref{def:permutation}, $\M' \models
  \sigma(\boxr{m}\psi)$.  The converse uses Zag and the inductive hypothesis.

  Suppose $\varphi = \neg \boxr{m} \psi$. Then $\M \models
  \neg \boxr{m} \psi$ iff there exists $\N \in R(m,w)$ such that, $\N \models
  \neg \psi$. Given
  that $\M\ssim \M'$, by Zig we know that for all $\N$ exist $\N'$
  such that $\N \ssim \N'$ and $\N' \in R'(\sigma(m),w')$. Then, by
  inductive hypothesis, $\N' \models \sigma(\neg \psi) = \neg
  \sigma(\psi)$ iff $\M' \models \neg \boxr{\sigma(m)}\sigma(\psi)$.
  Then, by Definition \ref{def:permutation}, $\M' \models
  \sigma(\neg \boxr{m}\psi)$.  The converse follows using Zag and the inductive hypothesis.
%
%
\end{proof}

An easily verifiable consequence of Definitions~\ref{def:sigma-model2}
and~\ref{def:sigma-sim} is that $\M$ and $\sigma(\M)$ are always
$\sigma$-similar.

\begin{proposition} \label{prop:lit1} Let $\sigma$ be a consistent
  permutation and $\mathcal{M} = \langle w,W,V,R \rangle$ a
  model. Then $\mathcal{M} \ssim \sigma (\mathcal{M})$.
\end{proposition}

\begin{proof}
  Let us define the relation $Z=\{ (\N,\sigma(\N)) \mid \N \in
  \Ext(\M) \}$ and show that it is a $\sigma$-simulation between $\M$
  and $\sigma(\M)$.  The \textit{Zig} and \textit{Zag} conditions are
  trivial by definition of $\sigma(\M)$.

  For \textit{Harmony}, we have to check that $l \in L_{V(w)}$
  iff $\sigma (l) \in L_{V'(w)}$. From the definition of $\sigma
  (\M)$, $L_{V'(w)} = \sigma (L_{V(w)})$, hence if $l \in L_{V(w)}$
  then $\sigma(l) \in \sigma(L_{V(w)})$.  Moreover, $\sigma
  (L_{V(w)})$ is a complete set of literals because $L_{V(w)}$ is a
  complete set of literals and $\sigma$ is a consistent permutation,
  and hence the converse also follows.
\end{proof}

Interestingly, if $\sigma$ is a symmetry of $\varphi$ then for any
model $\M$, $\M$ is a model of $\varphi$ if and only if $\sigma(\M)$
is. This will be a direct corollary of the following proposition in
the particular case when $\sigma$ is a symmetry and hence
$\sigma(\varphi) = \varphi$.

\begin{proposition}\label{prop:lit3}
  Let $\sigma$ be a consistent permutation, $\M$ a model and $\varphi$
  a modal CNF formula. Then $\M \models \varphi$ iff $\sigma (\M)
  \models \sigma (\varphi)$.
\end{proposition}

 \begin{proof}
   From Proposition~\ref{prop:lit1} ($\M
   \ssim \sigma(\M)$) and Proposition~\ref{prop:lit2}.
 \end{proof}

\begin{corollary}\label{cor:lit4}
  If $\sigma$ is a symmetry of $\varphi$ then $\M \in \Mod(\varphi)$
  iff $\sigma(\M) \in \Mod(\varphi)$.
\end{corollary}

To clarify the implications of the Corollary~\ref{cor:lit4}, consider
the following example.

\begin{example}
  Let $\varphi = (p \vee q \vee r) \wedge (s \vee q \vee r) \wedge
  (\neg p \vee \neg s) \wedge \tup{m} ( p \vee s ) \wedge [\A]( \neg r
  ) $.
  From Figure~\ref{fig:coroll}a we can verify
  that $\M_{1} \models \varphi$.

  Now $\sigma = (p\ s)(\neg p\ \neg s)$ is a symmetry of $\varphi$.
  Then, by Corollary~\ref{cor:lit4}, we have
  $\sigma(\M_{1}) \models \varphi$, which can be verified in the
  model of Figure~\ref{fig:coroll}b.

\begin{figure}[ht]
  \begin{center}
    \begin{minipage}[ht]{.45\linewidth}
      \begin{center}
        \includegraphics[scale=0.6,keepaspectratio=true]{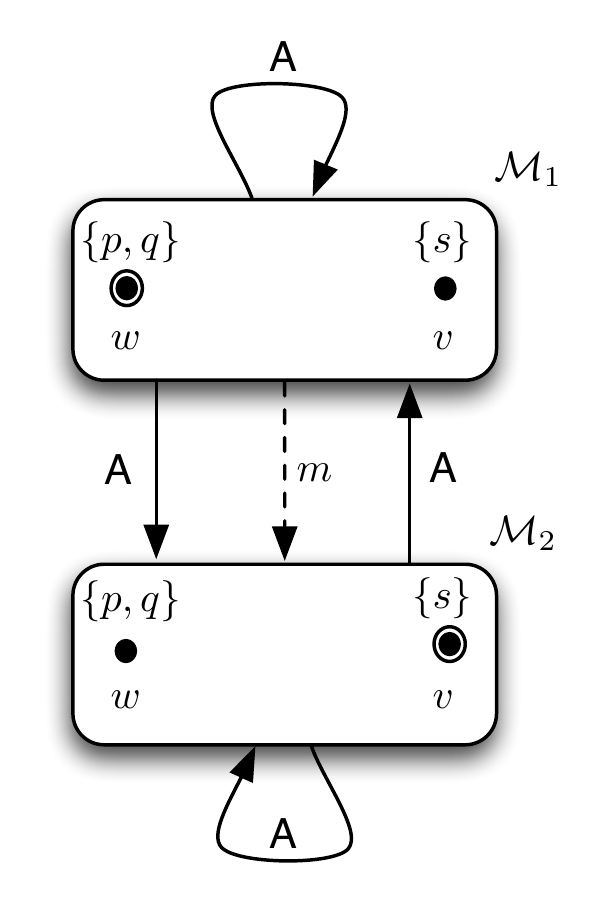}
          
        (a)
      \end{center}
    \end{minipage}
    \begin{minipage}[ht]{.45\linewidth}
      \begin{center}
        \includegraphics[scale=0.6,keepaspectratio=true]{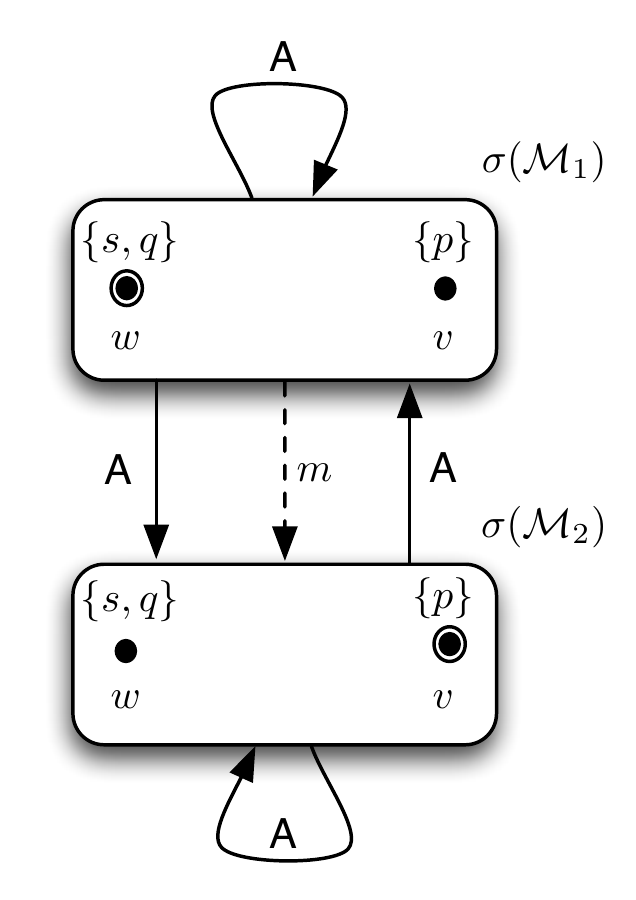}
          
        (b)
      \end{center}
    \end{minipage}
  \end{center}
  \caption{a) Model $\M_{1}$. b) Model $\sigma(\M_{1})$.  \label{fig:coroll}}
\end{figure}

\end{example}

The notion of $\sigma$-simulation in coinductive modal models is
general enough to be applicable to a wide range of modal logics.
Notice though, that our definition of $\sigma$-simulation makes no
assumption about the models being in the same class. Consider, for
example, a model $\M \in \cC_{\Hl(@)}$ and a permutation $\sigma = (i\
p)(\neg i \ \neg p)$ for $i \in \Nom$, $p \in \Prop$. By the defining
condition $\cC_{\Hl(@)}$, nominals in $\M$ are true at a unique
element in the domain, but this does not necessary hold for
$\sigma(\M)$, and hence $\sigma(\M)$ might not be in
$\cC_{\Hl(@)}$. Hence, when working with subclasses of $\Mods$ we will
often have to require additional conditions to a permutation $\sigma$
to ensure that for every $\M$, $\sigma(\M)$ is in the intended class.

\begin{definition}
  Let $\sigma$ be a permutation and $\cC$ a closed class of models. We
  say that $\cC$ is \emph{closed under $\sigma$} if for every $\M \in
  \cC$, $\sigma(\M) \in \cC$.
\end{definition}


\begin{example}[$\sigma$-simulation in hybrid logic]
  Consider the class $\cC_{\Hl(@)}$. $\cC_{\Hl(@)}$ is not closed under arbitrary permutations, but it is
  closed under permutations that send nominals to nominals.
%
%
%
\end{example}

Everything is now in place to show that modal entailment is preserved
under symmetries.

 \begin{theorem}\label{theo:main}
   Let $\varphi$ and $\psi$ be modal formulas, let $\sigma$ be a
   consistent symmetry of $\varphi$ and $\cC$ a class of models closed
   under $\sigma$. Then $\varphi \models_{\cC} \psi$ if and only if
   $\varphi \models_{\cC} \sigma(\psi)$.
 \end{theorem}

\begin{proof}
  We first show that under the hypothesis of the theorem the following
  property holds \smallskip

  \noindent
  \textbf{Claim:} $\Mod_{\cC}(\varphi) = \sigma(\Mod_{\cC}(\varphi))$.
  \smallskip

  \noindent [$\supseteq$] Let $\N \in
  \sigma(\Mod_{\cC}(\varphi))$ and $\M \in \Mod_{\cC}(\varphi)$ be
  such that $\N = \sigma (\M)$.  Then $\M \models \varphi$
  and by Corollary~\ref{cor:lit4}, $\sigma (\M) \models \varphi$. Given
  that $\cC$ is closed under $\sigma$, $\sigma (\M) \in \cC$ and,
  hence, $\sigma (\M) \in \Mod_{\cC}(\varphi)$.  \smallskip

  \noindent [$\subseteq$] Let $\M \in \Mod_{\cC}(\varphi)$, then $\M
  \models \varphi$.   By Corollary~\ref{cor:lit4}, $\sigma (\M) \models
  \varphi$ and, given that $\cC$ is closed under $\sigma$, $\sigma
  (\M) \in \cC$. Therefore, $\sigma (\M) \in
  \Mod_{\cC}(\varphi)$.  Because $\sigma$ is arbitrary, the results holds
  also for $\sigma^k$, $k \in \mathbb{Z}$. 
  
  Because $\sigma$ is a permutation over a finite set, there exists $n$ 
  such that $\sigma^n$ is the identity permutation. 
  Now consider $\sigma^{n-1}(\M)$, we know $\sigma^{n-1}(\M) \in
  \Mod_{\cC}(\varphi)$. Hence $\sigma^n(\M) = \M \in
  \sigma(\Mod_{\cC}(\varphi))$.\medskip

  Now, we have to prove that $\varphi \models_{\cC} \psi$ iff
  $\varphi \models_{\cC} \sigma(\psi)$.  By definition, $\varphi \models_{\cC}
  \psi$ iff $\Mod_{\cC}(\varphi) \models_{\cC} \psi$. By
  Proposition~\ref{prop:lit3}, this is the case if and only if
  $\sigma(\Mod_{\cC}(\varphi)) \models_{\cC} \sigma(\psi)$.

  Given that $\sigma$ is a symmetry of $\varphi$, by the Claim above,
  $\sigma(\Mod_{\cC}(\varphi)) \models_{\cC} \sigma(\psi)$ iff
  $\Mod_{\cC}(\varphi) \models_{\cC} \sigma(\psi)$, which by definition
  means $\varphi \models_{\cC} \sigma(\psi)$.
\end{proof}

Theorem~\ref{theo:main} provides an inexpensive inference mechanism
that can be used in every situation where entailment is involved
during modal automated reasoning.
Indeed, applying a permutation on a formula is a calculation that is
arguably computationally cheaper than a tableau expansion or a resolution
step. Therefore, new formulas obtained by this mean may reduce
the total running time of an inference algorithm.
In the case of propositional logic, the strengthening of the learning
mechanism has already shown its results in~\cite{Benhamou:uv}.
In the case of modal logic, it remains to see when cases of $\varphi \models \psi$
occur during a decision procedure, and how to better take advantage of them. 

\section{Layered Permutations} \label{sec:layering} 

In this section we
present the notion of layered permutations. First, we
present a definition of the tree model property~\cite{MLBOOK}
for coinductive modal models that we will use. 

Given a model $\M$, a (finite) \emph{path rooted at $\M$} is a
sequence $\pi=(\M_0,m_{1},\M_{1},\ldots,m_{k},\M_{k})$, for $m_i \in
\Mod$ where $\M_0 = \M$, $k \ge 0$, and $\M_{i} \in
R(m_{i},w^{\M_{i-1}})$ for $i=1,\ldots,k$. For a path
$\pi=(\M_0,m_{1},\M_{1},\ldots,m_{k},$ $\M_{k})$ we define $\first(\pi) =
\M_0$, $\last(\pi) = \M_k$, and $\length(\pi) = k$.  We denote the set
of all paths rooted at $\M$ as $\Pi[\M]$. A coinductive tree model is a model that has a unique path to every
reachable model (every model in $\Ext(\M)$).  Formally we can define
the class of all coinductive tree models, $\cC_{Tree}$, with the
following defining condition:
$$
\cC_{\Tree}:= P_{\Tree}(\M) \Longleftrightarrow \last:\Pi[\M] \mapsto
\Ext(\M) \mbox{ is bijective.}
$$

For example, the \emph{unravelling} construction (in its version for
coinductive modal models) shown below always defines a model in $\cC_{Tree}$.

\begin{definition}[Model Unravelling]
  Given a model $\M=\tup{w,W,V,R}$, the \emph{unravelling} of $\M$, (notation
  $\T(\M)$), is the rooted coinductive model $\T(\M) = \langle
  (\M),\Pi[\M],V',R'\rangle$ where
$$
\begin{array}{rcll}
  V'(\pi)  & = & V(w^{\last(\pi)}), \mbox{ for all } \pi \in \Pi[\M],\\
  R'(m, \pi)& = & \{ \tup{\pi', \Pi[\M],V',R'} \mid \last(\pi') \in
  R(m,\last(\pi))\},   & \mbox{for }m \in \Mod, \pi \in \Pi[\M].
\end{array}
$$
\end{definition}

It is easy to verify that given a model $\M$, its unravelling $\T(\M)$
is a tree ($\T(\M) \in \cC_{Tree}$) and, as expected, $\M$ and
$\T(\M)$ are bisimilar.

In what follows, we will use trees to define a more flexible family of
symmetries that we call \emph{layered symmetries}.  The following will
give a sufficient condition ensuring that layered symmetries also
preserve entailment.

\begin{definition}[Tree model closure property]\label{def:tree}
  We say that a class $\cC$ of models is \emph{closed under trees} if
  for every model $\M \in \cC$ there is a tree model $\T \in \cC$ such
  that $\M \bisim \T$.
\end{definition}

\noindent
From this definition, it follows that a class of models $\cC$ closed
under unravellings ($\T(\M) \in \cC$ for all $\M \in \cC$) is also
closed under trees.

\begin{example}
  Trivially the class $\Mods$ (i.e., the basic modal logic) is closed
  under trees, and so does the class
  $\cC_{\textsf{KAlt}_1}$ of models where the
  accessibility relation is a partial function.  Many classes like
  $\cC_\A$, $\cC_{@_i}$ and $\cC_\Nom$ fail to be closed under trees.
\end{example}

Logics defined over classes closed under trees have an interesting
property: there is a direct
correlation between the syntactical modal depth of the formula and the
depth in a tree model satisfying it.  In tree models, a notion of
layer is induced by the depth (distance from the root) of the nodes in
the model. Similarly, in modal formulas,  a notion of layer is induced 
by the nesting of the modal operators.  A consequence of this
correspondence is that literals occurring at different formula layers
are semantically independent of each other (see~\cite{arec:tree00} for
further discussion), i.e., at different layers the same literal can be
assigned a different value.

\begin{example}
Consider  the formula $\varphi = (p \vee q) \wedge (r \vee \neg \Box
(\neg p \vee q \vee \Box \neg r))$ and a tree model $\M$ of $\varphi$. Figure \ref{fig:layering} shows the layers induced by the
modal depth of the formula and the corresponding depth in $\M$. 
\begin{figure}[h]
  \begin{center}
    \includegraphics[scale=0.6,keepaspectratio=true]{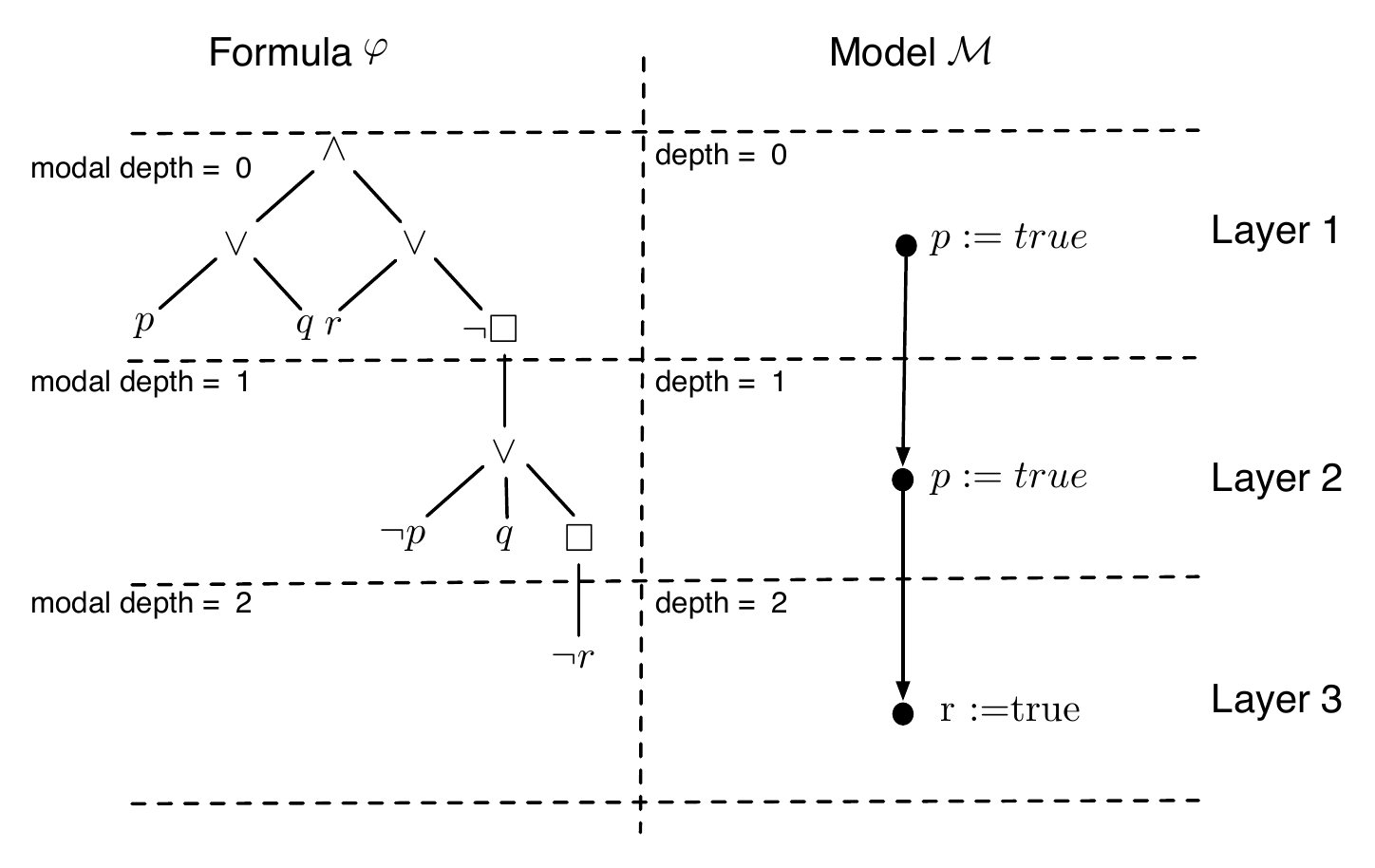}
  \end{center}
  \vspace*{-.3cm}
  \caption{Induced layering on a model and a formula.}\label{fig:layering}
\end{figure}

\end{example}

The independence between literals at different layers enables us to give a more flexible notion of
a permutation that we will call \emph{layered permutation}. 
Key to the notion of layered permutation is that of a \emph{permutation sequence}.

        	
\begin{definition}[Permutation Sequence]
  We define a finite permutation sequence $\bar{\sigma}$ as either
  $\bar{\sigma} = \tup{}$ (i.e., $\bar{\sigma}$ is the empty sequence)
  or $\bar{\sigma} = \sigma:\bar{\sigma}_2$ with $\sigma$ a
  permutation and $\bar{\sigma}_2$ a permutation sequence.
  Alternatively we can use the notation $\bar{\sigma} = \tup{\sigma_1,
    \ldots, \sigma_n}$ instead of $\bar{\sigma} =
  \sigma_1{:}\ldots{:}\sigma_n{:}\tup{}$.

  Let $|\bar{\sigma}|=n$ be the length of $\bar{\sigma}$ ($\tup{}$ has
  length 0).  For $1 \le i \le n$, we write $\bar{\sigma}_i$ for
  the subsequence that starts from the $i^{th}$ element of
  $\bar{\sigma}$. For $i \ge n$, we define $\bar{\sigma}_i =
  \tup{}$. In particular $\bar{\sigma} = \bar{\sigma}_1$. 
  Given a permutation sequence $\sigma_1:\bar{\sigma}_2$ we define
  $head(\sigma_1:\bar{\sigma}_2)=\sigma_1$ and $head(\tup{})=\sigma_{Id}$,
  where $\sigma_{Id}$ is the identity permutation.
  We say that a permutation sequence is \emph{consistent} if all of its 
  permutations are consistent.
\end{definition}

Applying a permutation sequence to a modal CNF formula can be
defined as follows:

\begin{definition}[Layered permutation of a formula]
  Let $\varphi$ be a modal CNF formula and $\bar{\sigma}$ a
  permutation sequence.  We define $\bar{\sigma}(\varphi)$
  recursively:
$$
\begin{array}{rcll}
  \tup{}(\varphi) & = & \varphi & \\ 
  (\sigma_1:\bar{\sigma}_2)(l) & = & \sigma_1(l) & \mbox{ for $l \in \ALit$} \\  
  (\sigma_1:\bar{\sigma}_2)(\boxr{m}C) & = & \boxr{\sigma_1(m)}\bar{\sigma}_2(C)\\
  \bar{\sigma}(C) & = & \csetsc{\bar{\sigma}(A)}{A \in C} & \mbox{ for $C$ a clause or a formula.}
\end{array}
$$
\end{definition}

Notice that layered permutations are well defined even if the modal depth
of the formula is greater than the size of the permutation sequence. 
Layered permutations let us use a different permutation at each modal depth. This
enables symmetries (layered symmetries) to be found, that would not be
found otherwise.
 
\begin{example}
  Consider the formula $\varphi = (p \vee \boxr{m} (p \vee \neg r))
  \wedge (\neg q \vee \boxr{m} (\neg p \vee r))$. If we only consider 
  non-layered symmetries then $\varphi$ has none. However,
  the permutation sequence $\tup{\sigma_1, \sigma_2}$ generated by
  $\sigma_1 = (p \ \neg q)$ and $\sigma_2 = (p \ \neg r)$ is a
  layered symmetry of $\varphi$.
\end{example}

As we can see from the previous example, layered permutations let us map 
the same literal to different targets at each different modal depth.  This 
additional degree of freedom can result in new symmetries for a given formula. 

From now on we can mostly repeat the work we did in the previous
section to arrive to a result similar to Theorem~\ref{theo:main} but
involving permutation sequences, with one caveat: the obvious
extension of the notion of permutated model $\sigma(\M)$ to layered
permutations is ill defined if $\M$ is not a tree.  Hence, we need the
additional requirement that the class $\cC$ of models is closed under
trees for the result to go through.

\begin{definition}[Layered Permutation of a model]\label{def:sigma-model}
  Let $\bar{\sigma}$ be a permutation sequence and $\M = \langle w, W, V, R \rangle$
  a tree model. Then $\bar{\sigma}(\M) = \langle w, W, V',R' \rangle$, where,
$$
\begin{array}{rll}
  V'(v) &=& head(\bar{\sigma}) (L_{V(v)}) \cap \Atom \quad \mbox{for all } v \in W
  \mbox{, and, }\\
  R'(m,v)&=&\{ \bar{\sigma}_{2}(\N) \mid \N \in R(head(\bar{\sigma})(m),v)\} \quad \mbox{for all
  } m \in \Mod \mbox{ and } v \in W.
\end{array}
$$
For $M$ a set of tree models, $\bar{\sigma}(M) = \{\bar{\sigma}(\M) \mid \M \in M\}$.
\end{definition}

We can now extend the notion of $\sigma$-simulation to
permutation sequences.

\begin{definition}[$\bar{\sigma}$-simulation]
  Let $\bar{\sigma}$ be a permutation sequence. A
  \emph{$\bar{\sigma}$-simulation} between models
  $\M=\tup{w,W,V,R}$ and $\M'=\tup{w',W',V',R'}$ is a family of
  relations $Z_{\bar{\sigma}_i} \subseteq \Ext(\M) \times \Ext(\M')$,
  $1 \le i $, that satisfies the following conditions:

  \begin{itemize}
  \item \textbf{Root:} $\M Z_{\bar{\sigma}_1} \M'$.

  \item \textbf{Harmony:} If $ wZ_{\bar{\sigma}_i}w'$ then $l
    \in L_{V(w)}$ iff $head(\bar{\sigma}_i)(l) \in L_{V'(w')}$.

  \item \textbf{Zig:} If $\M Z_{\bar{\sigma}_i}
    \M'$ and $\N \in R(m,w)$ then $\N Z_{\bar{\sigma}_{i+1}} \N'$ for
    some $\N' \in R'(head(\bar{\sigma}_i)(m),w')$.
    
  \item \textbf{Zag:} If $\M Z_{\bar{\sigma}_i}
    \M'$ and $\N' \in R'(m,w')$ then $\N Z_{\bar{\sigma}_{i+1}} \N'$
    for some $\N \in R(head(\bar{\sigma}_i)^{-1}(m),w)$.
  \end{itemize}
 
  We say that two models $\M$ and $\M'$ are $\bar{\sigma}$-similar
  (notation $\M \sssim \M')$, if there is a $\bar{\sigma}$-simulation
  between them.
\end{definition}

An important remark about the previous definition is that it does not make
any assumption about the size of the permutation sequence. In fact, it is
well defined even if the permutation sequence at hand is the empty
sequence. In that case, it just behave as the identity permutation at each
layer, thus the relation defines a bisimulation between the models.

Given a closed class of tree models $\cC$ and $\bar{\sigma}$ a
permutation sequence, we say that $\cC$ is \emph{closed under
  $\bar{\sigma}$} if for every $\M \in \cC$, $\bar{\sigma} (\M) \in
\cC$.

Now we are ready to prove the main result concerning layered
symmetries and entailment.

\begin{theorem}\label{theo:main2}
  Let $\varphi$ and $\psi$ be modal formulas and let $\bar{\sigma}$ be
  a consistent permutation sequence, and
  let $\cC$ be a class of models closed under trees and $\cC \cap \cC_{Tree}$ closed under
  $\bar{\sigma}$.  If $\bar{\sigma}$ is a symmetry of $\varphi$ then
  for any $\psi$ we have that $\varphi \models_\cC \psi$ if and only if
  $\varphi \models_\cC \bar{\sigma}(\psi)$.
\end{theorem}

\begin{proof}
  We first show that under the hypothesis of the theorem the following
  two properties hold. \smallskip

  \noindent
  \textit{Claim 1:} $\Mod_{\cC \cap \cC_{Tree}}(\varphi) =
  \bar{\sigma}(\Mod_{\cC \cap \cC_{Tree}}(\varphi))$.
  
  \smallskip The argument is the same as for the Claim in Theorem
  \ref{theo:main} but using permutation sequences.  \medskip

\noindent
\textit{Claim 2:} $\Mod_{\cC}(\varphi) \models_{\cC} \varphi \mbox{
  iff } \Mod_{\cC \cap \cC_{Tree}}(\varphi) \models_{\cC} \varphi$.

\smallskip The left-to-right direction is trivial by the fact that
$\Mod_{\cC \cap \cC_{Tree}}(\varphi) \subseteq \Mod_{\cC}(\varphi)$.
For the other direction, assume $\Mod_{\cC \cap \cC_{Tree}}(\varphi)
\models_{\cC} \varphi$ and $\Mod_{\cC}(\varphi) \not \models_{\cC}
\varphi$. Then there is $\M \in \Mod_{\cC}(\varphi)$ such that $\M
\not \models_{\cC} \varphi$. But we know that $\M \bisim \T$, and
$\T \in \Mod_{\cC \cap \cC_{Tree}}(\varphi)$. Hence $\T
\models_{\cC} \varphi$ which contradicts our assumption.

\medskip It rests to prove that $\varphi \models_{\cC} \psi$ if
and only if $\varphi \models_{\cC} \bar{\sigma}(\psi)$.  By
definition, $\varphi \models_{\cC} \psi$ if and only if
$\Mod_{\cC}(\varphi) \models_{\cC} \psi$. By Claim 2, this is case if
and only if $\Mod_{\cC \cap \cC_{Tree}}(\varphi) \models_{\cC} \psi$.
By the layered version of Proposition~\ref{prop:lit3}, this is
the case if and only if $\bar{\sigma}(\Mod_{\cC \cap
  \cC_{Tree}}(\varphi)) \models_{\cC} \bar{\sigma}(\psi)$.  Given that
$\bar{\sigma}$ is a symmetry of $\varphi$, by Claim 1,
$\bar{\sigma}(\Mod_{\cC \cap \cC_{Tree}}(\varphi)) \models_{\cC}
\bar{\sigma}(\psi)$ if and only if $\Mod_{\cC \cap
  \cC_{Tree}}(\varphi) \models_{\cC} \bar{\sigma}(\psi)$, which by
Claim 2 is the case if and only if $\Mod_{\cC}(\varphi) \models_{\cC}
\bar{\sigma}(\psi)$ which by definition means that $\varphi
\models_{\cC} \bar{\sigma}(\psi)$.
\end{proof}



\newcommand{\cname}{C_{m,k,i}}
\newcommand{\sext}{\sigma_{ext}}

\section{Symmetry Detection} \label{sec:detection}

Different techniques have been proposed for detecting symmetries of
propositional formulas in clausal form. Some of them, deal directly
with the formula~\cite{Benhamou:1994tr}, while others, reduce the
problem to the problem of finding automorphisms in colored graphs
constructed in such a way that the automorphism group of the graph is
isomorphic to the symmetry group of the formula under
consideration~\cite{Crawford:1992wz,Crawford:1996wa,Aloul:2002ww}.

The availability of efficient tools to detect graph automorphisms
(e.g.,~\cite{McKay:1990vf,Darga:2004us,Junttila:2007vx}) has made the
later approach the most successful one because it is fast and easy to
integrate.  


In this section we present a technique for the detection of symmetries
in modal formulas that extends the construction proposed for
propositional formulas to modal CNF
formulas. We present the graph construction algorithm and prove its
correctness. 

We now introduce some notation and definitions.  In what follows, we
consider modal CNF formulas as set of sets as defined in Section~\ref{sec:defs} and write $\psi \in \varphi$ to express that
$\psi$ is subformula  of $\varphi$. Clauses occurring at modal depth 0 are named \emph{top
  clauses} and clauses occurring in modal literals are named
\emph{modal clauses}.
Let $s:\Mod \times \{0,1\} \mapsto \mathbb{N}\backslash \{0,1\}$ be an
injective function and let $t: Sub(\varphi) \mapsto \mathbb{N}$ be a
partial function defined as:
$$
t(\psi) = \left \{
  \begin{array}{lll}
    1 & \mbox{if } \psi \mbox{ is a top clause } \\
    s(m,0) & \mbox{if } \psi = \boxr{m}C \\
    s(m,1) & \mbox{if } \psi = \neg \boxr{m}C
  \end{array}
\right.
$$

\noindent
The \emph{typing} function $t$ assigns a numeric type to every clause
(top or modal). For modal clauses, the type is based on the modality
and the polarity of the modal literal in which it occurs.
%
%
Let us assign to each clause $C$ occurring in $\varphi$ a unique identifier $\mathit{id}(C)={\tup{m, k, i}}$ where $m$ is the modal depth at which the clause occurs, $k=t(\psi)$ is the type of the clause as returned by the
\emph{typing} function $t$ and $i \in \mathbb{N}$ is different for each clause. To 
simplify notation, in what follows we will assume that each clause $C$ is labeled by 
its unique identifier $\mathit{id}(C)={\tup{m, k, i}}$ and write $C_{m,k,i}$. 

By definition, a symmetry of a formula $\varphi$ is a bijective function that maps
literals to literals. It can naturally be extended to a function $\sext$ that also maps 
each clause $C$ to $\sigma(C)$.  Notice that because
$\sigma$ is a symmetry of $\varphi$ both $C$ and $\sigma(C)$ are clauses in $\varphi$. 
Hence, both $C$ and $\sigma(C)$ will be assigned some identifier by the $\mathit{id}$ function. 
 
%

The following are properties of $\sext$ that are easy to verify. 

\begin{proposition}\label{prop:sext}
  Let $\varphi$ be a modal CNF formula and $\sigma$ a symmetry of
  $\varphi$. Then for the extension of $\sigma$, $\sext$, the
  following holds:
  \begin{enumerate}[label=\emph{\roman*})]
  \item\label{sext:1} $\sext$ is a bijective function.
  \item\label{sext:3} If $\sext(\cname) = C_{m',k',i'}$ then $m=m'$.
  \item\label{sext:4} If $\sext(\cname) = C_{m',k',i'}$ then $k=k'$.
  \item\label{sext:5} If $l \in \cname$ then $\sext(l) \in \sext(\cname)$.
  \item\label{sext:2} $\sext$ is a symmetry of $\varphi$.
    
  \end{enumerate}
\end{proposition}

\newcommand{\Gv}{G(\varphi)}

We can now introduce the construction of the colored graph corresponding
to a given formula $\varphi$. We will construct an undirected colored graph
with two types of edges. 
As coloring function we will use the typing
function $t$ introduced earlier in this section.


\begin{definition}\label{def:graph}
  Let $\varphi$ be a modal CNF formula and let $At(\varphi)$ denote
  the set of atoms occurring in $\varphi$. The colored graph
  $\Gv=(V,E_{1},E_{2})$ is constructed as follows:

  \begin{enumerate}
  \item For each atom $a \in At(\varphi)$:

    \begin{enumerate}
    \item Add two \emph{literal} nodes of color 0: one labelled $a$
      and one labelled $\neg a$.

    \item Add an edge to $E_{1}$ between these two nodes to ensure Boolean
      consistency.
    \end{enumerate}

   \item For each top clause $C$ of $\varphi$ add a \emph{clause} node of
     color $t(C)$. 

   \item \label{cons:point2} For each atom literal occurring in $C$,
     add an edge to $E_{1}$ from $C$ to the corresponding literal node.

  \item For each modal literal $\boxr{m}C'$ ($\neg \boxr{m}C'$) occurring in $C$:
    \begin{enumerate}
    \item Add a clause node of color $t(\boxr{m}C')$
      ($t(\neg\boxr{m}C')$) to represent the modal clause $C'$.

    \item Add an edge to $E_{1}$ from the node of $C$ to this node.
    \item If $m$ is indexed by an atom literal $l$ then add an edge to $E_{2}$ 
      from the node $C'$ to the indexing literal $l$.

    \item Repeat the process from point \ref{cons:point2} for each
      literal (atom or modal) occurring in $C'$.
    \end{enumerate}

  \end{enumerate}

\end{definition}

This construction creates a graph with $2 + 2|\Mod|$ colours and at
most $(2|V| +\#(\mbox{TopClauses}) + \#(\mbox{ModalClauses}))$ nodes.

\begin{example}
  Let us consider the following  formula $\varphi = (a \vee \boxr{m}(b
  \vee \neg \boxr{m}c)) \wedge (b \vee \boxr{m}(a  \vee \neg \boxr{m}c))$.
  This formula has six clauses (2 at modal depth 0, 2 at modal depth 1
  and 2 at modal depth 2) and three atoms (six literals).  The
  associated colored graph, $\Gv$, is shown in Figure~\ref{fig:graph1}
  (colors are represented by shapes in the figure).

  \begin{figure}[ht]
    \begin{center}
      \begin{tabular}{ll}
        \begin{picture}(200,0)
          \put(0,-50){\includegraphics[scale=0.5,keepaspectratio=true]{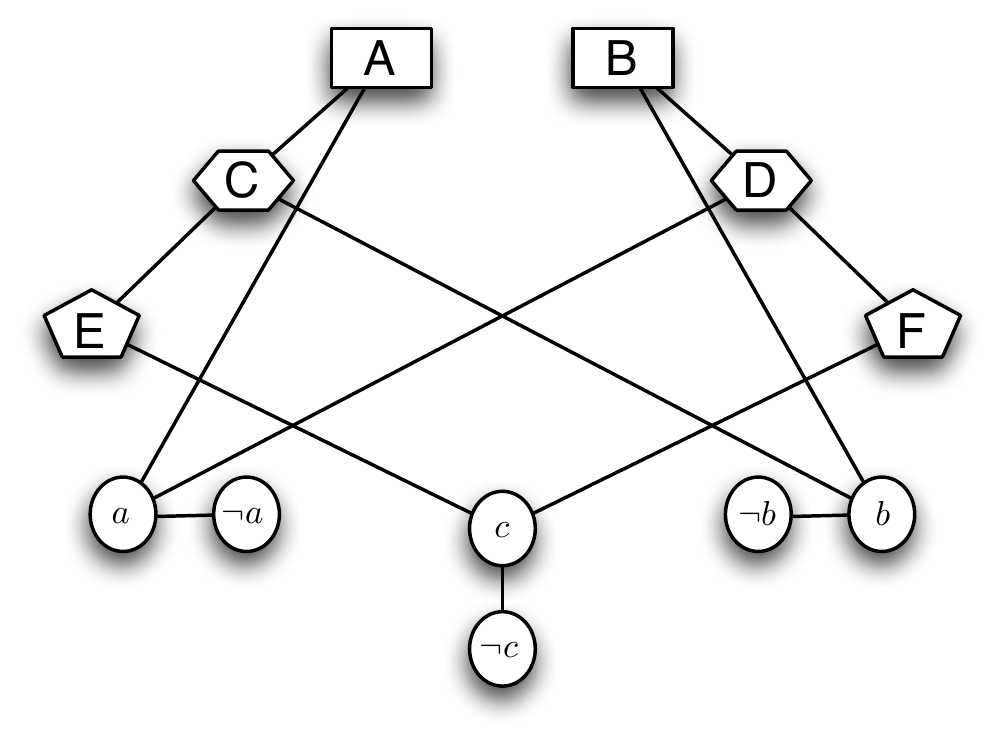}}
        \end{picture}
        &

        $\begin{array}{l}
          A =  (a \vee \boxr{m}(b \vee \neg \boxr{m}c))\\
          B =  (b \vee \boxr{m}(a \vee \neg \boxr{m}c))\\
          C = \boxr{m}(b \vee \neg \boxr{m}c)\\
          D = \boxr{m}(a \vee \neg \boxr{m}c) \\
          E = \boxr{m}c \\
          F = \boxr{m}c\\
        \end{array}$
      \end{tabular}
    \end{center}
    \caption{Graph representation of $\varphi$.}\label{fig:graph1}
  \end{figure}

\end{example}

Note that the construction of Definition~\ref{def:graph} induces a
mapping $g$ that associates to each literal and clause the
corresponding node in the graph.  

To prove that the proposed construction is correct, we first have to
show that each symmetry of the formula is a colored automorphism of
the graph.

\begin{proposition}\label{prop:sym2graph}
  Let $\varphi$ be a modal CNF formula, $\sigma$ a symmetry of
  $\varphi$, $\Gv=(V,E_1,E_2)$ the colored graph of $\varphi$ as
  defined by Definition \ref{def:graph}, and $g$ the mapping induced
  by the construction of $\Gv$. Then $\pi= g \circ \sext$ is an
  automorphism of $\Gv$.
\end{proposition}

\begin{proof}
To simplify notation lets assume that $g$ is
the identity function (i.e., we do not differentiate between a clause
(or literal) and its associated node in the graph) and as a consequence $\pi = \sext$.
Then $\pi$ is an automorphism of  $\Gv$ if the following holds:

\begin{enumerate}

\item \label{aut:1}$(l,\neg l) \in E_{1}$
  iff $(\pi(l),\pi(\neg l)) \in E_{1}$ for all $l \in V$.

  We have to consider the following cases:
  \begin{itemize}
    \setlength{\itemsep}{5pt}
  \item[-] $\sigma$ is a permutational symmetry: \\
    \underline{($\rightarrow$):} 
    Assume $\pi=\sext=(a \ b)(\neg a\ \neg
    b)$. Then $(\pi(a),\pi(\neg a)) = (b,\neg b) \in E_1$
    by construction of $\Gv$.
      
    \underline{($\leftarrow$):} $(a,\neg a) \in E_1$ by construction of $\Gv$.  

  \item[-] $\sigma$ is a phase-shift symmetry:\\
    \underline{($\rightarrow$):} Assume $\pi=\sext=(a \ \neg
    a)$. Then $(\pi(a),\pi(\neg a)) = (\neg a, a)$, but given
    that $\Gv$ is an undirected graph $(\neg a, a) = (a,\neg a) \in
    E_{1}$ by construction.
    
    \underline{($\leftarrow$):} $(a,\neg a) \in E_{1}$ by construction of $\Gv$.

  \item[-] $\sigma$ is a compositional symmetry: It follows
    directly from the previous two cases.
 \end{itemize}

\item \label{aut:2}$(l, \cname) \in E_{1}$ iff $(\pi(l),
  \pi(\cname)) \in E_{1}$ for all $l, \cname \in V$. 

\underline{($\rightarrow$):} By construction, $(l, \cname) \in E_{1}$ only
if $l \in \cname$. Then, given that $\sext$ is a symmetry and
by Proposition \ref{prop:sext}\ref{sext:5}, we 
know that $\sext(l)$ and $\sext(\cname)$ both occur in $\varphi$ and
$\sext(l) \in \sext(\cname)$. Then, by construction, we have
that $(\pi(l), \pi(\cname)) \in E_{1}$.

\underline{($\leftarrow$):} It follows directly by construction of
$\Gv$.

\item \label{aut:3}$(C_{m,k,i},C_{m',k',i'}) \in E_{1}$ iff
  $(\pi(C_{m,k,i}),\pi(C_{m',k',i'})) \in
  E_{1}$ for all $C_{m,k,i},C_{m',k',i'} \in V$.

  \underline{($\rightarrow$):} If $(C_{m,k,i},C_{m',k',i'}) \in E_{1}$ we
  know that either $m < m'$ or $m > m'$. Assume $m < m'$ then
  $C_{m',k',i'}$ is a modal clause occurring in $C_{m,k,i}$. By
  Proposition \ref{prop:sext}\ref{sext:5}, we have that
  $\sext(C_{m',k',i'})$ is a modal clause occurring in
  $\sext(C_{m,k,i})$, and given that $\sext$ is a symmetry of
  $\varphi$, $\sext(C_{m,k,i})$ and $\sext(C_{m',k',i'})$ both occur
  in $\varphi$, therefore, by construction, $\pi(C_{m,k,i}) \in V$ and,
  $\pi(C_{m',k',i'}) \in V$, therefore $(\pi(C_{m,k,i}),\pi(C_{m',k',i'})) \in E_{1}$.

\underline{($\leftarrow$):} It follows directly by construction of
$\Gv$.

\item \label{aut:4}$(l,\cname) \in E_{2}$ iff
  $(\pi(l),\pi(\cname)) \in E_{2}$ for all $l, \cname \in V $.

  \underline{($\rightarrow$):} $(l,\cname) \in E_{2}$ if the modality
  of the modal clause $\cname$ is indexed by $l$. Given that $\sext$
  is a symmetry of $\varphi$, we know that $\sext(l) \in
  \varphi$ and
  $\sext(\cname) \in \varphi$ and that $\sext(l)$ index the
  modality of the modal clause $\sext(\cname)$, therefore, by
  construction, $\pi(l) \in V$ and $\pi(\cname) \in V$ and
  $(\pi(l),\pi(\cname)) \in E_{2}$.

\underline{($\leftarrow$):} It follows directly by construction of
$\Gv$.

\item \label{aut:5} For every cycle $(x \
  y) \in \pi$, $x$ and $y$ have the same color.

Follows from Proposition~\ref{prop:sext}\ref{sext:4} and the fact that
by construction different types of clauses are assigned different
colors in the graph.
\end{enumerate}
\end{proof}

We now prove that any colored automorphism of $\Gv$ induces a symmetry of $\varphi$.

\begin{proposition}\label{prop:graph2sym}
  Let $\varphi$ be a modal CNF formula,  $\Gv=(V,E_1,E_2)$ the colored graph of $\varphi$ as defined
  by Definition \ref{def:graph}, $\pi$ an automorphism of
  $\Gv$ and $g$ the mapping induced by the construction of
  $\Gv$. Then $\sext= g^{-1} \circ \pi$ is a symmetry of
  $\varphi$.
\end{proposition}

\begin{proof}
  Once more, assume $g$ is the identity. To prove that $\sext$ is a symmetry of $\varphi$, we have to prove
  the following properties:

  \begin{enumerate}
  \item $\sext$ is a consistent permutation, i.e., $\sext(\neg l)
    = \neg \sext(l)$ for all $l \in \ALit$.

By construction, Boolean consistency edges only connect
literal nodes. Let $l_i \in V$ be a literal node. Then by construction
we have that $(l_i, \neg l_i) \in E_1$. Now assume that $\pi(l_i)=l_j$
for $l_j \in V$. Given
that $\pi$ is an automorphism it must be the case that
$(\pi(l_i), \pi(\neg l_i)) \in E_1$, and therefore that
$\pi(\neg l_i)=\neg l_j = \neg \pi(l_i)$, which implies that
$\sext(\neg l_i)=\neg \sext(l_i)$.

\item If $\cname \in \varphi$ then $\sext(\cname) \in \varphi$.

  By construction $\cname \in V$ implies that $\cname \in \varphi$. As
  $\pi$ is an automorphism, $\pi(\cname) \in V$, therefore,
  $\pi(\cname) \in \varphi$ which implies that $\sext(\cname) \in \varphi$.

\item If $l \in \varphi$ then $\sext(l) \in \varphi$.

It follows by the same argument as in the previous case.

\item If $\sext(\cname)=C_{m',k',i'}$ then $k=k'$.

It follows from the fact that $\pi$ is a colored automorphism,
mapping only nodes of the same color, and that by construction, clauses of
the same type are assigned the same color in the
graph.

\item If $\sext(\cname)=C_{m',k',i'}$ then $m=m'$.

We prove this by induction on $m$, the modal depth at which a
clause occurs in $\varphi$.

\underline{Base Case: $m=0$.}
We have to prove that if $\sext(C_{0,k,i})=C_{m',k',i'}$ then $m'=0$.
Assume that $m'\neq 0$. Then, there is a clause $C_{n,s,j}$, with $n<m'$ such
that, $C_{m',k',i'}$ is a modal clause occurring in it. By
construction, we then have that $(C_{n,s,j},C_{m',k',i'}) \in E_{1}$. As
$\pi$ is an automorphism of $\Gv$, we should have
$(\pi(C_{n,s,j}),\pi(C_{m',k',i'})) = (\pi(C_{n,s,j}),C_{0,k,i}) \in
E_{1}$, but by construction there is no such edge.

\underline{Inductive Step: $n<m \implies m$.}
By construction of $\Gv$ if $(C_{m,k,i},C_{n,l,j}) \in E_{1}$ then $|m - n| = 1$.
Now, assume $m\neq m'$. We know that there is a clause
$C_{(m'-1),s,j}$ such that $(C_{(m'-1),s,j},C_{m',k',i'}) \in
E_{1}$.
Then, as $\pi$ is an automorphism of $\Gv$, it must be the case that
$(\pi(C_{(m'-1),s,j}),\pi(C_{m',k',i'}))=(\pi(C_{(m'-1),s,j}),C_{m,k,i})
\in E_{1}$. By the inductive hypothesis we know that
$\pi(C_{(m'-1),s,j}) = C_{(m'-1),s,j'}$ and therefore, we have that
$(C_{(m'-1),s,j'},C_{m,k,i}) \in E_{1}$. But then we get that $|m -
(m'-1)| \geq 2$, which by construction cannot happen. Therefore
 $(C_{(m'-1),s,j'},C_{m,k,i}) \not \in E_{1}$,
contradicting the fact that $\pi$ is an automorphism of $\Gv$.

\item If $l$ index a clause $\cname$ then $\sext(l)$ index the
  clause $\sext(\cname)$.

If $l$ index a clause $\cname$, then by construction $(l,\cname) \in
E_2$. Given that $\pi$ is an automorphism, $(\pi(l),\pi(\cname)) \in
E_2$, which implies that $\sext(l)$ index the
  clause $\sext(\cname)$.
  \end{enumerate}
We have proved that $\sext$, the extension of $\sigma$, obtained
from an automorphism of the graph, is a symmetry of $\varphi$. To obtain
the original symmetry $\sigma$ we just take the restriction of $\sext$ to
atom literals.
\end{proof}

Finally we can prove that our construction is correct. 
\begin{theorem}
  Let $\varphi$ be a modal CNF formula and $\Gv= (V,E_{1},E_{2})$ the
  colored graph constructed following the construction of Definition
  \ref{def:graph}. Then every symmetry $\sigma$ of $\varphi$
  corresponds one-to-one to an automorphism $\pi$ of
  $\Gv$.
\end{theorem}
\begin{proof}
Immediate from Proposition~\ref{prop:sym2graph} and~\ref{prop:graph2sym}.
\end{proof}

This construction enables the detection of symmetries as defined
in Section~\ref{sec:lit-symm},
that is, symmetries defined over literals that can appear at various
modal depths of a given formula.
To detect layered symmetries (see Section~\ref{sec:layering}) for logics with the tree model closure
property,  the construction needs to be
modified to capture the notion of layers. This is easy to achieve by just
changing the way literals occurring in the formula are handled
(see~\cite{orbe2012} for details). Properties~\ref{prop:sym2graph} and~\ref{prop:graph2sym} 
(suitably generalized) also hold in this case.

\section{Conclusions and Further Work}\label{sec:conclusion}

The notion of symmetry has been well studied in propositional logic,
and various optimizations of decision procedures based on it are
known.  In this article, we extend the notion of syntactic and
semantic symmetries to many different modal logics using the framework
of coinductive models.  The main contribution is that a 
symmetry $\sigma$ preserves entailments whenever the class is closed
by $\sigma$.  For example, arbitrary  symmetries preserve
entailments in the basic modal logic, but for the hybrid logic
$\mathcal{H}(@)$ we can only consider symmetries that map nominals to
nominals.  The second contribution of the paper is to show that if the
class of models is closed under trees, then the more flexible notion
of layered symmetry also preserve entailments.

To arrive at the previous results, we defined the concept of
$\sigma$-simulation and showed that it preserves $\sigma$-permutation
of formulas.  We then presented permutation sequences $\bar{\sigma}$,
and $\bar{\sigma}$-simulations. Permutation sequences are relevant in
those classes of models that are closed under trees.  This property
enables the use of layered symmetries, a notion that can capture more
symmetries than the ordinary symmetry definition. Indeed, layered
symmetries can be detected independently within atoms at each modal
depth of a formula.  $\bar{\sigma}$-simulations extend the notion of
$\sigma$-simulations to permutation sequences, and enabled us to prove
that layered symmetries also preserve entailment.

Finally, we presented a method to detect symmetries in modal formulas
that reduces the problem to the graph automorphism problem. Given a
formula $\varphi$, the idea is to build a graph in such a way that the
automorphism group of the graph is isomorphic to the symmetry group of
the formula that generated it. A general construction algorithm, suitable for
many modal logics, was presented and its correctness proved. 
The presented graph construction algorithm can be extended to detect
layered symmetries. Preliminary results on modal symmetries
concerning this last construction can be found in~\cite{orbe2012}
where we developed an efficient algorithm to detect symmetries for the
basic modal logic, and empirically verified that many modal problems (both randomly and hand generated) contain
symmetries.


Our ongoing research focuses on the incorporation of symmetry
information into a modal tableau calculi such
as~\cite{KaSmoJoLLI2009,hoffmann2010tab} or modal resolution calculi
such as~\cite{arec:reso08}.  
One promising theme that we will investigate in the future is
permutations involving also modal literals.


\bigskip
\noindent
\textbf{Acknowledgments.}
 This work was partially supported by grants ANPCyT-PICT-2008-306, ANPCyT-PICT-2010-688, the FP7-PEOPLE-2011-IRSES Project
``Mobility between Europe and Argentina applying Logics to Systems'' (MEALS)
and the Laboratoire International Associ\'e ``INFINIS''.

\bibliographystyle{eptcs}
\bibliography{references}


\end{document}